\documentclass[letterpaper, 10 pt]{ieeeconf}  
 
\IEEEoverridecommandlockouts 
\usepackage{geometry}
\newgeometry{top=0.75in, bottom=0.75in, right=0.75in, left=0.75in}

\newcommand{\zono}[1]{\langle #1 \rangle}
\usepackage{packagearticleieee}
\usepackage{amssymb}

\newtheorem{theorem}{Theorem}
\newtheorem{proposition}{Proposition}

\usepackage{cite}
\usepackage{booktabs}
\usepackage{multirow}

 \usepackage{algorithm} 
 \usepackage{algpseudocode}
 \usepackage[utf8]{inputenc}
\usepackage{graphicx}
\usepackage{subcaption}

\makeatletter
\let\NAT@parse\undefined
\makeatother
\usepackage{url}
\usepackage[colorlinks=true,linkcolor=blue!80!black,citecolor=blue!80!black,urlcolor=blue!80!black]{hyperref}
\DeclareMathSymbol{\shortminus}{\mathbin}{AMSa}{"39}

\def\tr{\text{tr}}

\title{\LARGE\bf Transformer-Enhanced Data-Driven Output Reachability \\ with Conformal Coverage Guarantees}

\author{Zhen Zhang, Peng Xie, Wenyuan Wu, Yanliang Huang, and Amr Alanwar
\thanks{All authors are with the School of Computation, Information and Technology, Technical University of Munich, Germany. (Email: $\{$zhenzhang.zhang, p.xie, yanliang.huang,  wenyuan.wu, alanwar$\}$@tum.de)}
}

\begin{document}
\maketitle
\thispagestyle{empty}
\pagestyle{empty}

\begin{abstract}
This paper considers output reachability analysis for linear time-invariant systems with unknown state-space matrices and unknown observation map, given only noisy input-output measurements. The Cayley--Hamilton theorem is applied to eliminate the latent state algebraically, producing an autoregressive input-output model whose parameter uncertainty is enclosed in a matrix zonotope. Set-valued propagation of this model yields output reachable sets with deterministic containment guarantees under a bounded aggregated residual assumption. The conservatism inherent in the lifted matrix-zonotope product is then mitigated by a decoder-only Transformer trained on labels obtained through directional contraction of the formal envelope via an exterior non-reachability certificate. Split conformal prediction restores distribution-free coverage at both per-step and trajectory levels without access to the true reachable-set hull. The framework is validated on a five-dimensional system with multiple unknown observation matrices.
\end{abstract}

\section{Introduction}

Reachability analysis computes the set of states or outputs a dynamical system can attain from given initial conditions under bounded inputs and disturbances, and underpins collision avoidance~\cite{althoff2008reachability}, safe motion planning~\cite{althoff2021set}, and formal verification~\cite{rierson2017developing}. Classical set-propagation methods~\cite{girard2005reachability,girard2006efficient} assume known dynamics, yet accurate models are often unavailable. This has motivated data-driven reachability analysis, where safety certificates are constructed directly from measured data.

Existing data-driven methods focus on state-space models. Matrix-zonotope-based approaches characterize all dynamics consistent with noisy state measurements and propagate reachable sets with formal guarantees~\cite{alanwar2021data,althoff2010reachability}; constrained matrix zonotopes further reduce conservatism~\cite{alanwar2023data}. For input-output models, reachability analysis of ARMAX systems has been developed using dependency-preserving symbolic zonotopes that mitigate the wrapping effect inherent in standard set arithmetic~\cite{luetzow2023reachability}; however, these methods require a pre-identified ARMAX model with known noise structure rather than operating directly from raw data with bounded-but-unknown residuals. Probabilistic reachable sets have been studied through kernel embeddings~\cite{thorpe2021sreachtools}, Christoffel-function level sets~\cite{devonport2023data}, Gaussian processes~\cite{griffioen2023data}, deep neural approximations~\cite{sivaramakrishnan2024stochastic}, and conformal inference~\cite{hashemi2024statistical}. Conformalized data-driven reachability with PAC guarantees~\cite{huang2026cddr} combines neural surrogates with the Learn-Then-Test calibration procedure, providing stronger trajectory-level coverage but requiring state-level access or known system structure for model-set construction. Data-driven set-based estimation~\cite{alanwar2022data} is closely related. Classical subspace identification methods such as N4SID~\cite{van1994n4sid} can also operate from input-output data, but require specific noise structure assumptions and do not provide set-valued reachability guarantees. All set-based methods above assume a known output matrix: in reachability~\cite{alanwar2021data,alanwar2023data} it recovers the latent state; in estimation~\cite{alanwar2022data} it maps outputs back to states. When the output matrix is also unknown, neither mechanism remains available, yet this setting is practically important whenever only noisy input-output data are accessible. Data-driven output reachability without knowledge of any system matrix remains an open problem.

This paper addresses this gap by applying the Cayley--Hamilton theorem~\cite{horn2012matrix} to derive an autoregressive input-output model that eliminates the latent state. The unknown parameter set is constructed as a matrix zonotope~\cite{althoff2010reachability} from offline data, and output reachable sets are propagated via zonotopic set operations. The framework requires no knowledge of any system matrix and relies only on a sufficiently exciting input-output trajectory, the system order, and a valid residual bound.

A key challenge is conservatism. In state-space data-driven reachability~\cite{alanwar2023data}, the regressor dimension equals the sum of the state and input dimensions and the over-approximation error remains moderate. The autoregressive formulation lifts the regressor dimension in proportion to the system order, causing the data-consistent parameter set to grow accordingly. Because this growth originates from the parameter set itself, standard remedies are ineffective: zonotope order reduction~\cite{girard2005reachability} and step-size refinement~\cite{kuhn1998rigorously} act only on the propagated set; richer representations such as constrained zonotopes~\cite{scott2016constrained} or constrained polynomial zonotopes~\cite{kochdumper2020sparse} cannot shrink the matrix zonotope encoding all data-consistent dynamics; constrained matrix zonotope refinement~\cite{alanwar2023data} and exact set-valued multiplication~\cite{zhang2025data} reduce the wrapping effect but cannot alter the underlying parameter set size. The conservatism is therefore structural. In set-based estimation, observations can mitigate conservatism~\cite{alanwar2022data}, but such observations are unavailable in reachability prediction. To address this, we train a decoder-only Transformer~\cite{vaswani2017attention} on certificate-assisted tightened labels to replace multi-step propagation with autoregressive set-valued prediction, and apply split conformal prediction~\cite{vovk2005algorithmic,shafer2008tutorial} to obtain distribution-free finite-sample coverage guarantees for realized output trajectory points.

The main contributions of this paper are as follows. First, we develop a data-driven framework for output reachability that eliminates the latent state via an autoregressive input-output representation, enabling reachable-set computation directly from noisy input-output data without knowledge of any system matrix. Second, we show how external non-reachability certificates, when available during training, can be exploited to generate tighter labels within the formal data-driven envelope, and we use these labels to train a decoder-only Transformer for output-set prediction~\cite{vaswani2017attention}. Third, we combine the learned predictor with split conformal prediction~\cite{vovk2005algorithmic} to provide distribution-free finite-sample coverage guarantees for realized output trajectories, both per-step and jointly over the prediction horizon.

The paper is organized as follows. Section~\ref{sec:preliminaries} introduces notation and formulates the problem. Section~\ref{sec:output-reach} presents the data-driven output reachability framework. Section~\ref{sec:tf-reach} develops the Transformer-enhanced approach with conformal calibration. Section~\ref{sec:experiments} reports numerical results, and Section~\ref{sec:conclusion} concludes.

\section{Preliminaries and Problem Formulation}\label{sec:preliminaries}

\subsection{Notations}
The sets of real and natural numbers are denoted by $\R$ and $\N$, respectively. 
We denote the matrix of zeros and ones of size $m \times n$ by $0_{m \times n}$ and $1_{m \times n}$, respectively, and the identity matrix by $I_n\in\R^{n\times n}$. Subscripts are omitted when dimensions are clear from context. 
For a matrix $A$, $A^\top$ denotes the transpose and $A^\dagger$ the Moore--Penrose pseudoinverse. For a vector $v$, $\|v\|_2$ and $\|v\|_\infty$ denote the Euclidean and infinity norms, respectively. The operator $\mathrm{proj}_{1:n_y}(\mathcal{S})$ denotes the projection of a set $\mathcal{S}$ onto its first $n_y$ coordinates. The probability of an event is denoted by $\Pr[\cdot]$.

\subsection{Set Representations}

\begin{definition}[Zonotope {\cite{kuhn1998rigorously}}] \label{def:zonotopes}
Given a center $c_{\mathcal{Z}} \in \mathbb{R}^{n_x}$ and a generator matrix
$G_{\mathcal{Z}}=\begin{bmatrix} g_{\mathcal{Z}}^{(1)} & \cdots & g_{\mathcal{Z}}^{(\gamma_{\mathcal{Z}})} \end{bmatrix}
\in \mathbb{R}^{n_x \times \gamma_{\mathcal{Z}}}$ with $\gamma_{\mathcal{Z}} \in \mathbb{N}$, the zonotope associated with $c_{\mathcal{Z}}$ and $G_{\mathcal{Z}}$ is defined by
\begin{equation}
\mathcal{Z} =
\Bigl\{
x \in \mathbb{R}^{n_x}
\;\Big|\;
x \!= \!c_{\mathcal{Z}} \!+ \sum_{i=1}^{\gamma_{\mathcal{Z}}} \alpha^{(i)} g_{\mathcal{Z}}^{(i)},
\alpha^{(i)} \!\in\![\shortminus 1,\!1]
\Bigr\}.
\end{equation}
We compactly write $\mathcal{Z}=\zono{c_{\mathcal{Z}},G_{\mathcal{Z}}}$.
\end{definition}

For zonotopes $\mathcal{Z}_1=\zono{c_{\mathcal{Z}_1},G_{\mathcal{Z}_1}}$ and
$\mathcal{Z}_2=\zono{c_{\mathcal{Z}_2},G_{\mathcal{Z}_2}}$, and a matrix
$L \in \mathbb{R}^{m \times n_x}$, standard operations remain in zonotopic form~\cite{althoff2010reachability}: the image under a linear map is
$L\mathcal{Z}_1=\zono{Lc_{\mathcal{Z}_1},LG_{\mathcal{Z}_1}}$; the Minkowski sum by $\mathcal{Z}_1 \oplus \mathcal{Z}_2=\zono{c_{\mathcal{Z}_1}+c_{\mathcal{Z}_2},[G_{\mathcal{Z}_1},G_{\mathcal{Z}_2}]}$, the Minkowski difference by $\mathcal{Z}_1 \ominus \mathcal{Z}_2 = \{z_1 \in \mathbb{R}^n \mid z_1 \oplus \mathcal{Z}_2 \subseteq \mathcal{Z}_1\}$; 
% the Minkowski sum is
% $\mathcal{Z}_1+\mathcal{Z}_2=\zono{c_{\mathcal{Z}_1}+c_{\mathcal{Z}_2},[G_{\mathcal{Z}_1},G_{\mathcal{Z}_2}]}$; the Minkowski difference is
% $\mathcal{Z}_1-\mathcal{Z}_2=\zono{c_{\mathcal{Z}_1}-c_{\mathcal{Z}_2},[G_{\mathcal{Z}_1},-G_{\mathcal{Z}_2}]}$; and the Cartesian product is
$\mathcal{Z}_1 \times \mathcal{Z}_2=\zono{\begin{bmatrix} c_{\mathcal{Z}_1} \\ c_{\mathcal{Z}_2} \end{bmatrix},\mathrm{blkdiag}(G_{\mathcal{Z}_1},G_{\mathcal{Z}_2})}$.
We use $+$ to denote the Minkowski sum when unambiguous from context. Likewise, $\mathcal{Z}_1-\mathcal{Z}_2$ denotes $\mathcal{Z}_1+(-\mathcal{Z}_2)$, not the Minkowski difference.

\begin{definition}[Matrix Zonotope {\cite[p.~52]{althoff2010reachability}}] \label{def:matzonotopes}
Given a center matrix $C_{\mathcal{M}} \in \mathbb{R}^{n_x \times p}$ and $\gamma_{\mathcal{M}} \in \mathbb{N}$ generator matrices $\tilde{G}_{\mathcal{M}}=\begin{bmatrix} G_{\mathcal{M}}^{(1)}&\dots&G_{\mathcal{M}}^{(\gamma_{\mathcal{M}})} \end{bmatrix} \in \mathbb{R}^{n_x \times (p \gamma_{\mathcal{M}})}$, a matrix zonotope is defined as
%\cite[Def.~1]{Girard2005} 
\begin{equation}
	\mathcal{M} {=} \Big\{\! X \in \mathbb{R}^{n_x \times p}   \Big| X {=} C_{\mathcal{M}} \!+ \!\sum_{i=1}^{\gamma_{\mathcal{M}}} \alpha^{(i)}  G_{\mathcal{M}}^{(i)} ,\!
	\alpha^{(i)} \!\in\![\shortminus 1,\!1]\! \Big\} .
\end{equation}
We use the shorthand notation $\mathcal{M} = \zono{C_{\mathcal{M}},\tilde{G}_{\mathcal{M}}}$ for a matrix zonotope. %\hfill $\square$
\end{definition}

\subsection{System Model and Assumptions}

We consider a discrete-time linear time-invariant system with unknown dynamics:
% \begin{align}\label{eq:ss}
%     x_{(k+1)} &= A \, x_{(k)} + B \, u_{(k)} + w_{(k)}, \nonumber \\
%     y_{(k)}   &= C \, x_{(k)} + v_{(k)},
% \end{align}
\begin{equation}\label{eq:ss}
\begin{aligned}
    x_{(k+1)} &= A \, x_{(k)} + B \, u_{(k)} + w_{(k)}, \\
    y_{(k)}   &= C \, x_{(k)} + v_{(k)}.
\end{aligned}
\end{equation}
where $x_{(k)} \in \R^{n_x}$ is the state, $u_{(k)} \in \R^{n_u}$ the input, $y_{(k)} \in \R^{n_y}$ the output, $w_{(k)} \in \R^{n_x}$ the process noise, and $v_{(k)} \in \R^{n_y}$ the measurement noise. The matrices $A$, $B$, $C$ are unknown; only input-output data $\{(u_{(k)}, y_{(k)})\}$ are available.

\begin{assumption}\label{ass:order}
    The system order $n_o = n_x$ is known. The pair $(C, A)$ is observable.
    \hfill $\lrcorner$
\end{assumption}

\begin{remark}
The system order in Assumption~\ref{ass:order} can be estimated via SVD of Hankel matrices; observability ensures the output contains sufficient information for the autoregressive representation in Section~\ref{sec:output-reach}.
\hfill $\lrcorner$
\end{remark}

\subsection{Problem Statement}

We start by introducing the exact output reachable set of system~\eqref{eq:ss}.

\begin{definition}[Exact output reachable set]\label{def:output-reach-set}
The output reachable set of system~\eqref{eq:ss} at step $k$ is
\begin{align}\label{eq:output-reach}
    \mathcal{Y}_k = \big\{ y_{(k)} \in \R^{n_y} \mid & x_{(j+1)} = A x_{(j)} + B u_{(j)} + w_{(j)}, \nonumber \\&
    y_{(j)} = C x_{(j)} + v_{(j)}, 
    x_{(0)} \in \mathcal{X}_0,  \nonumber \\& u_{(j)} \in \mathcal{U}_j, 
    j = 0, \dots, k \big\},
\end{align}
where $\mathcal{X}_0$ and $\mathcal{U}_j$ denote the initial state set and the input set, respectively.
\hfill $\lrcorner$
\end{definition}

The goal is to compute an over-approximation $\hat{\mathcal{Y}}_k \supseteq \mathcal{Y}_k$ using only measured input-output data, without knowledge of $A$, $B$, $C$ and without reconstructing $x_{(k)}$.

\subsection{Conformal Prediction}

Split conformal prediction~\cite{vovk2005algorithmic} provides distribution-free coverage guarantees. Given a calibration set $\{(X_i, Y_i)\}_{i=1}^{n_\text{cal}}$ of exchangeable pairs, a predictor $\hat{f}$, and a miscoverage level $\delta \in (0,1)$, nonconformity scores $s_i = \|Y_i - \hat{f}(X_i)\|$ are computed. The conformal quantile is
\begin{equation}\label{eq:conformal}
    \hat{q} = \text{Quantile}\Big(\{s_i\}_{i=1}^{n_\text{cal}}, \; \frac{\lceil (n_\text{cal}+1)(1-\delta) \rceil}{n_\text{cal}}\Big).
\end{equation}
where $\mathrm{Quantile}(\mathcal{S}, \tau)$ denotes the $\lceil \tau |\mathcal{S}| \rceil$-th smallest element of a finite set $\mathcal{S}$.
The prediction set $\mathcal{C}(X) = \{y : \|y - \hat{f}(X)\| \leq \hat{q}\}$ satisfies
$\Pr[Y_{n_\text{cal}+1} \in \mathcal{C}(X_{n_\text{cal}+1})] \geq 1 - \delta$.

\section{Data-Driven Output Reachability}\label{sec:output-reach}

\subsection{Output Autoregressive Representation}

The key observation enabling output-only reachability is that the state $x_{(k)}$, which is latent since only input-output data are available, can be eliminated algebraically from system~\eqref{eq:ss} by the Cayley--Hamilton theorem~\cite{horn2012matrix}, yielding the following autoregressive input-output representation.

\begin{lemma}[Output autoregressive form {\cite[Ch.~7]{hespanha2018linear}}]\label{lm:armax}
Consider system~\eqref{eq:ss} under Assumption~\ref{ass:order}. For all $k \geq n_o$ with $n_o = n_x$, the output satisfies
\begin{align}\label{eq:armax}
    y_{(k)} &= -\sum_{i=1}^{n_o} a_i \, y_{(k-i)} + \sum_{i=1}^{n_o} b_i \, u_{(k-i)} + \varepsilon_{(k)},
\end{align}
where $a_1, \dots, a_{n_o}$ are the coefficients of the characteristic polynomial of $A$ with $a_0 = 1$, the input coefficients are $b_i = \sum_{j=0}^{i-1} a_j \, C A^{i-1-j} B$, and the aggregated residual is
\begin{equation}\label{eq:noise_agg}
    \varepsilon_{(k)} = \sum_{i=0}^{n_o - 1}\! \Big(\sum_{j=0}^{i} a_j \, C A^{i-j}\Big) w_{(k-1-i)}  + v_{(k)} + \sum_{i=1}^{n_o} a_i \, v_{(k-i)}.
\end{equation}
\end{lemma}

% Proof removed per standard result; see \cite[Ch.~4]{ljung1999system}.
% \begin{proof}
% Iterating the state equation in~\eqref{eq:ss} from step $k{-}n_o$ to step $k$ yields
% $x_{(k)} = A^{n_o} x_{(k-n_o)} + \sum_{i=0}^{n_o-1} A^i (B\,u_{(k-1-i)} + w_{(k-1-i)})$.
% Left-multiplying by $C$ and adding $v_{(k)}$ gives
% \begin{align*}
%     y_{(k)} = C A^{n_o} x_{(k-n_o)} &+ \sum_{i=0}^{n_o-1} C A^i B\,u_{(k-1-i)} \\
%     &+ \sum_{i=0}^{n_o-1} C A^i w_{(k-1-i)} + v_{(k)}.
% \end{align*}
% By Lemma~\ref{lm:cayley-hamilton}, $C A^{n_o} x_{(k-n_o)} = -\sum_{j=1}^{n_o} a_j\, C A^{n_o-j} x_{(k-n_o)}$.
% For each $j \in \{1,\dots,n_o\}$, iterating the state equation from $k{-}n_o$ to $k{-}j$ gives $x_{(k-j)} = A^{n_o-j} x_{(k-n_o)} + \text{(input/noise terms)}$, so
% $C A^{n_o-j} x_{(k-n_o)} = C x_{(k-j)} - C\sum_{\ell=0}^{n_o-j-1} A^\ell(Bu_{(k-j-1-\ell)} + w_{(k-j-1-\ell)})$.
% Since $Cx_{(k-j)} = y_{(k-j)} - v_{(k-j)}$ by the output equation in~\eqref{eq:ss}, substituting and collecting all output, input, and noise contributions produces~\eqref{eq:armax}--\eqref{eq:noise_agg}. No inversion of $C$ is required; the state $x$ is eliminated entirely through the algebraic identity $Cx_{(k-j)} = y_{(k-j)} - v_{(k-j)}$.
% \end{proof}

\begin{assumption}\label{ass:residual-bound}
There exists a known zonotope $\mathcal{Z}_\varepsilon$ containing the origin such that $\varepsilon_{(k)}\in\mathcal{Z}_\varepsilon$ for all $k\ge n_o$.
\hfill $\lrcorner$
\end{assumption}

\begin{remark}\label{rem:residual-bound}
Bounded-uncertainty assumptions are standard in set-valued estimation and reachability~\cite{girard2005reachability,girard2006efficient,alanwar2021data,alanwar2023data,milanese2004set,combastel2015zonotopes}. Assumption~\ref{ass:residual-bound} is strictly weaker than separately bounding $\mathcal{Z}_w$ and $\mathcal{Z}_v$~\cite{alanwar2021data,alanwar2023data,alanwar2022data}, since $\mathcal{Z}_\varepsilon$ aggregates both noise channels into a single input-output residual without requiring $A$, $B$, or $C$. In practice, $\mathcal{Z}_\varepsilon$ can be estimated from a held-out trajectory by computing the empirical residual range and inflating by a safety margin. A conservative choice affects only tightness, not correctness of Proposition~\ref{prop:reach}. The conformal guarantee in Theorem~\ref{thm:coverage} is independent of $\mathcal{Z}_\varepsilon$ entirely.
\hfill $\lrcorner$
\end{remark}

To enable autoregressive propagation in the output space, we introduce the lifted observation vector
\begin{equation}\label{eq:lifted}
    z_{(k)}^\top = \begin{bmatrix} y_{(k-1)}^\top & \cdots & y_{(k-n_o)}^\top & u_{(k-1)}^\top & \cdots & u_{(k-n_o)}^\top \end{bmatrix}
\end{equation}
in $\R^{n_o(n_y + n_u)}$. By Lemma~\ref{lm:armax}, the output $y_{(k)}$ is a linear function of $z_{(k)}$ and the current input $u_{(k)}$ up to the aggregated residual $\varepsilon_{(k)}$. The next lifted vector $z_{(k+1)}$ is obtained by shifting $z_{(k)}$ and appending $y_{(k)}$ and $u_{(k)}$, so that the transition
\begin{equation}\label{eq:lifted-transition}
    z_{(k+1)} = \Theta^\top \begin{bmatrix} z_{(k)} \\ u_{(k)} \end{bmatrix} + \varepsilon_{(k)}
\end{equation}
holds for an unknown parameter matrix $\Theta^\top \in \R^{n_o(n_y+n_u) \times (n_o(n_y+n_u)+n_u)}$ and a bounded residual vector $\varepsilon_{(k)} \in \mathcal{Z}_\varepsilon$. The first $n_y$ rows of~\eqref{eq:lifted-transition} recover $y_{(k)}$ from Lemma~\ref{lm:armax}, while the remaining rows encode the deterministic shift structure of the observation window; accordingly, only the first $n_y$ components of $\varepsilon_{(k)}$ are nonzero, and the matrix zonotope $\mzon_\varepsilon$ in~\eqref{eq:Cme}--\eqref{eq:Gme} can be restricted to this block to reduce conservatism in the model set.

\subsection{Data-Driven Model Set Construction}

We collect $K$ input-output trajectories. For the $i$-th trajectory, starting from index $n_o$ so that the initial-condition transient has decayed, the lifted observation and augmented regressor data matrices are
\begin{align*}
    Z_+ = \begin{bmatrix} z^{(1)}_{(n_o+1)} & \! \cdots \! & z^{(K)}_{(T_K+1)} \end{bmatrix}, 
    \Phi_- = \begin{bmatrix} z^{(1)}_{(n_o)} &\!  \cdots \! & z^{(K)}_{(T_K)} \\ u^{(1)}_{(n_o)} & \! \cdots\!  & u^{(K)}_{(T_K)} \end{bmatrix}.
\end{align*}
Let $T$ denote the total number of data columns. The stacked residual matrix $E_- \in \R^{n_o(n_y+n_u) \times T}$ belongs to a matrix zonotope $\mzon_\varepsilon = \zono{C_{\mzon_\varepsilon}, \tilde{G}_{\mzon_\varepsilon}}$, constructed by concatenating $\mathcal{Z}_\varepsilon$ across all $T$ columns~\cite{alanwar2023data}:
\begin{align}
    C_{\mzon_\varepsilon}
    &= \begin{bmatrix}
        c_{\mathcal{Z}_\varepsilon} & \cdots & c_{\mathcal{Z}_\varepsilon}
    \end{bmatrix}, \label{eq:Cme} \\
    G_{\mzon_\varepsilon}^{(j+(i-1)T)}
    &= \begin{bmatrix}
        0_{p \times (j-1)} &
        g_{\mathcal{Z}_\varepsilon}^{(i)} &
        0_{p \times (T-j)}
    \end{bmatrix}, \label{eq:Gme}
\end{align}
for all \(i \in \{1,\dots,\gamma_{\mathcal{Z}_\varepsilon}\}\) and
\(j \in \{1,\dots,T\}\), where \(p = n_o(n_y+n_u)\).
The consistent parameter set is
\begin{equation}\label{eq:Nsig}
    \mathcal{N}_\Sigma = \big\{ \Theta^\top \mid Z_+ = \Theta^\top \Phi_- + E_-, \; E_- \in \mzon_\varepsilon \big\}.
\end{equation}

\begin{assumption}\label{ass:rank}
    The augmented regressor data matrix $\Phi_- \in \R^{(n_o(n_y+n_u)+n_u) \times T}$ has full row rank, i.e., $\mathrm{rank}(\Phi_-) = n_o(n_y + n_u) + n_u$.
    \hfill $\lrcorner$
\end{assumption}

\begin{lemma}\label{lm:model_set}
Under Assumptions~\ref{ass:residual-bound}--\ref{ass:rank}, the matrix zonotope
\begin{equation}\label{eq:Msigma}
    \mzon_\Sigma = (Z_+ - \mzon_\varepsilon) \, \Phi_-^\dagger
\end{equation}
satisfies $\Theta_\tr^\top \in \mzon_\Sigma$ and $\mathcal{N}_\Sigma \subseteq \mzon_\Sigma$.
\end{lemma}

\begin{proof}
The argument extends~\cite[Lemma~1]{alanwar2023data} from the state-space to the output autoregressive setting.
For any $\Theta^\top \in \mathcal{N}_\Sigma$, there exists $E_- \in \mzon_\varepsilon$ such that $\Theta^\top \Phi_- = Z_+ - E_-$. Since $\Phi_-$ has full row rank by Assumption~\ref{ass:rank}, right-multiplying by $\Phi_-^\dagger$ gives $\Theta^\top = (Z_+ - E_-) \Phi_-^\dagger \in \mzon_\Sigma$.
\end{proof}

\subsection{Output Reachable Set Propagation}

Given initial output zonotopes $\hat{\mathcal{Y}}_0, \dots, \hat{\mathcal{Y}}_{n_o-1}$ and the model set $\mzon_\Sigma$, the lifted observation zonotope $\hat{\mathcal{Z}}_k$ and the output reachable set are propagated jointly.

\begin{proposition}\label{prop:reach}
Under Assumptions~\ref{ass:residual-bound}--\ref{ass:rank}, for $k \geq n_o$,
\begin{equation}\label{eq:reach_prop}
    \hat{\mathcal{Z}}_{k+1} = \mzon_\Sigma \, (\hat{\mathcal{Z}}_k \times \mathcal{U}_k) \oplus \mathcal{Z}_\varepsilon,
\end{equation}
and the output reachable set $\hat{\mathcal{Y}}_k$ is obtained by projecting $\hat{\mathcal{Z}}_{k+1}$ onto the first $n_y$ coordinates. The containment $\mathcal{Y}_k \subseteq \hat{\mathcal{Y}}_k$ holds for all $k \geq n_o$.
\end{proposition}

\begin{proof}
Since $\Theta_\tr^\top \in \mzon_\Sigma$ by Lemma~\ref{lm:model_set} and $[z_{(k)}^\top, u_{(k)}^\top]^\top \in \hat{\mathcal{Z}}_k \times \mathcal{U}_k$ by construction, the true lifted vector $z_{(k+1)} = \Theta_\tr^\top [z_{(k)}; u_{(k)}] + \varepsilon_{(k)}$ belongs to $\mzon_\Sigma \, (\hat{\mathcal{Z}}_k \times \mathcal{U}_k) \oplus \mathcal{Z}_\varepsilon$ for all admissible inputs and noise. The output $y_{(k)}$ occupies the first $n_y$ entries of $z_{(k+1)}$, so projecting yields $\mathcal{Y}_k \subseteq \hat{\mathcal{Y}}_k$.
\end{proof}

\begin{corollary}[Trajectory-level containment]\label{thm:init-indep}
Let $\hat{\mathcal{Y}}_0, \dots, \hat{\mathcal{Y}}_{n_o-1}$ be initial output zonotopes used to construct the initial lifted set $\hat{\mathcal{Z}}_{n_o}$. Suppose that for every realized trajectory $\{y_{(j)}\}_{j=0}^{n_o-1}$ of the true system~\eqref{eq:ss} with $u_{(j)} \in \mathcal{U}_j$, the individual measurements satisfy $y_{(j)} \in \hat{\mathcal{Y}}_j$ for $j = 0, \dots, n_o{-}1$. Then, the data-driven output reachable set from Proposition~\ref{prop:reach} satisfies $y_{(k)} \in \hat{\mathcal{Y}}_k$ for all $k \geq n_o$ along that trajectory.
\end{corollary}

\begin{proof}
Fix a trajectory of the true system. By assumption, $y_{(j)} \in \hat{\mathcal{Y}}_j$ for $j = 0, \dots, n_o{-}1$ and $u_{(j)} \in \mathcal{U}_j$. The initial lifted vector $z_{(n_o)} = [y_{(n_o-1)}^\top, \dots, y_{(0)}^\top, u_{(n_o-1)}^\top, \dots, u_{(0)}^\top]^\top$ therefore satisfies $z_{(n_o)} \in \hat{\mathcal{Z}}_{n_o} = \hat{\mathcal{Y}}_{n_o-1} \times \cdots \times \hat{\mathcal{Y}}_0 \times \mathcal{U}_{n_o-1} \times \cdots \times \mathcal{U}_0$. Since $\Theta_{\tr}^\top \in \mzon_\Sigma$ by Lemma~\ref{lm:model_set} and $\varepsilon_{(k)} \in \mathcal{Z}_\varepsilon$, the propagation~\eqref{eq:reach_prop} gives $z_{(k+1)} \in \hat{\mathcal{Z}}_{k+1}$ for all $k \geq n_o$ by induction. Projecting onto the first $n_y$ coordinates yields $y_{(k)} \in \hat{\mathcal{Y}}_k$.
\end{proof}

The procedure is summarized in Algorithm~\ref{alg:dd-reach}.

\begin{algorithm}[t]
\caption{Data-Driven Output Reachability}
\label{alg:dd-reach}
\textbf{Input:} $(Z_+, \Phi_-)$, initial lifted set $\hat{\mathcal{Z}}_{n_o}$, noise $\mzon_\varepsilon$, input sets $\mathcal{U}_k$, horizon $N$, order $\rho_{\max}$\\
\textbf{Output:} $\hat{\mathcal{Y}}_{n_o}, \dots, \hat{\mathcal{Y}}_N$
\begin{algorithmic}[1]
\State $\mzon_\Sigma \gets (Z_+ - \mzon_\varepsilon) \, \Phi_-^\dagger$ 
\For{$k = n_o, \dots, N$}
    \State $\hat{\mathcal{Z}}_{k+1} \gets \mzon_\Sigma \, (\hat{\mathcal{Z}}_k \times \mathcal{U}_k) \oplus \mathcal{Z}_\varepsilon$
    \State $\hat{\mathcal{Y}}_k \gets \mathrm{proj}_{1:n_y}(\hat{\mathcal{Z}}_{k+1})$
    % \State $\hat{\mathcal{Z}}_{k+1} \gets \operator{reduce}(\hat{\mathcal{Z}}_{k+1}, \rho_{\max})$
\EndFor
\end{algorithmic}
\end{algorithm}

\begin{remark}\label{rem:complexity}
The autoregressive lifting raises the regressor dimension to $n_o(n_y{+}n_u){+}n_u$, substantially exceeding the $n_x{+}n_u$ of the state-space setting~\cite{alanwar2023data}. The resulting over-approximation in $\mzon_\Sigma \, (\hat{\mathcal{Z}}_k \times \mathcal{U}_k)$ is structural and cannot be reduced by zonotope order reduction or step-size refinement, motivating the learning-based surrogate in Section~\ref{sec:tf-reach}.
\hfill $\lrcorner$
\end{remark}

\section{Transformer-Enhanced Output Reachability}\label{sec:tf-reach}

Although Algorithm~\ref{alg:dd-reach} provides formal containment, the autoregressive lifting inflates the regressor dimension well beyond the state-space case~\cite{alanwar2023data}, and the correspondingly larger matrix zonotope $\mzon_\Sigma$ introduces substantial over-approximation that neither the reduce operator nor step-size refinement can mitigate. Moreover, the required initial lifted sets $\hat{\mathcal{Z}}_0,\dots,\hat{\mathcal{Z}}_{n_o-1}$ are unavailable in the purely data-driven setting. These limitations motivate the learning-based surrogate below.

\subsection{Zonotope Fitting from Trajectory Samples}

We first address the initialization problem. Given a finite point set $\{p^{(t)}\}_{t=1}^M \subset \R^{n_y}$, let $\bar{c} = \frac{1}{M}\sum_{t=1}^M p^{(t)}$ denote the sample mean and $P \in \R^{M \times n_y}$ the centered data matrix with rows $(p^{(t)}-\bar{c})^\top$. Let $P = V \Sigma U^\top$ be its singular value decomposition with right singular vectors $u_1,\dots,u_{n_y}$ (the columns of $U$).

\begin{lemma}[PCA-oriented containing zonotope]\label{lm:pca-fit}
Given $\{p^{(t)}\}_{t=1}^M \subset \R^{n_y}$ with sample mean $\bar{c}$ and right singular vectors $u_1,\dots,u_{n_y}$ as above, define the directional radii
\begin{equation}\label{eq:pca-radius}
\rho_i = \max_{1 \le t \le M} \bigl|u_i^\top (p^{(t)}-\bar{c})\bigr|, \quad i=1,\dots,n_y,
\end{equation}
and the generator matrix $G_{\mathrm{fit}} = [\rho_1 u_1, \dots, \rho_{n_y} u_{n_y}] \in \R^{n_y \times n_y}$. Then the zonotope $\mathcal{Z}_{\mathrm{fit}} = \zono{\bar{c},\, G_{\mathrm{fit}}}$ satisfies $\{p^{(t)}\}_{t=1}^M \subset \mathcal{Z}_{\mathrm{fit}}$.
\end{lemma}

\begin{proof}
Since $U_r=[u_1,\dots,u_{n_y}]$ is an orthonormal basis of $\R^{n_y}$, each centered sample admits the expansion
$
p^{(t)}-\bar{c} = \sum_{i=1}^{n_y} \bigl(u_i^\top (p^{(t)}-\bar{c})\bigr) u_i.
$
For each $t$ and $i$, define
$
\alpha_i^{(t)} =
\begin{cases}
\dfrac{u_i^\top (p^{(t)}-\bar{c})}{\rho_i}, & \rho_i > 0,\\[1.2ex]
0, & \rho_i = 0.
\end{cases}
$
By the definition of $\rho_i$ in~\eqref{eq:pca-radius}, one has $|\alpha_i^{(t)}|\le 1$ for all $i=1,\dots,n_y$ and all $t=1,\dots,M$. Hence
$
p^{(t)}-\bar{c}
= \sum_{i=1}^{n_y} \alpha_i^{(t)} \rho_i u_i
= G_{\mathrm{fit}} \alpha^{(t)},
$
where $\alpha^{(t)} = [\alpha_1^{(t)},\dots,\alpha_{n_y}^{(t)}]^\top$ satisfies
$\|\alpha^{(t)}\|_\infty \le 1$. Therefore,
$
p^{(t)} \in \zono{\bar{c},\,G_{\mathrm{fit}}} = \mathcal{Z}_{\mathrm{fit}}
$
for every $t=1,\dots,M$.
\end{proof}

% \begin{lemma}[PCA zonotope fitting]\label{lm:pca-fit}
% Let $\{p^{(t)}\}_{t=1}^M \subset \R^{n_y}$ be a finite point set with sample mean $\bar{c} = \frac{1}{M}\sum_{t=1}^M p^{(t)}$. Let $V \Sigma U^\top$ be the singular value decomposition of the centered data matrix $P = [p^{(1)}{-}\bar{c}, \dots, p^{(M)}{-}\bar{c}]^\top \in \R^{M \times n_y}$, and let $v_1, \dots, v_{n_y}$ denote the right singular vectors (columns of $U$). Define the zonotope $\mathcal{Z}_{\mathrm{fit}} = \zono{\bar{c}, G_{\mathrm{fit}}}$ with generators
% \begin{equation}\label{eq:pca-gen}
%     g_i = v_i \cdot \max_{1 \leq t \leq M} |v_i^\top (p^{(t)} - \bar{c})|, \quad i = 1, \dots, K_g,
% \end{equation}
% where $K_g \leq n_y$. Then $\{p^{(t)}\}_{t=1}^M \subset \mathcal{Z}_{\mathrm{fit}}$.
% \end{lemma}

% \begin{proof}
% For any $p^{(t)}$, writing $p^{(t)} - \bar{c} = \sum_{i=1}^{n_y} (v_i^\top (p^{(t)} - \bar{c})) \, v_i$ and defining $\alpha_i^{(t)} = v_i^\top (p^{(t)} - \bar{c}) / \max_{s} |v_i^\top (p^{(s)} - \bar{c})|$, we have $|\alpha_i^{(t)}| \leq 1$ and $p^{(t)} = \bar{c} + \sum_{i=1}^{K_g} \alpha_i^{(t)} g_i + r^{(t)}$, where $r^{(t)}$ lies in the subspace spanned by $v_{K_g+1}, \dots, v_{n_y}$. When $K_g = n_y$, the residual vanishes and $p^{(t)} = \bar{c} + G_{\mathrm{fit}} \alpha^{(t)}$ with $\|\alpha^{(t)}\|_\infty \leq 1$, hence $p^{(t)} \in \mathcal{Z}_{\mathrm{fit}}$.
% \end{proof}

For the first $n_o$ steps, Lemma~\ref{lm:pca-fit} is applied with $p^{(t)} = y_{(k)}^{(t)}$ to obtain the initial context zonotopes.

\subsection{Certificate-Assisted Label Refinement}\label{sec:tightening}

The data-driven sets from Proposition~\ref{prop:reach} are guaranteed outer approximations but may be excessively conservative. When exterior information about the true reachable set is available during training (e.g., from a model-based oracle), it can be exploited to tighten the training labels. This section describes a directional contraction procedure that produces zonotopes $\mathcal{T}_k \subseteq \hat{\mathcal{Y}}_k^{\mathrm{DD}}$ as Transformer training targets.

\begin{assumption}[Directional exterior certificate]\label{ass:directional-certificate}
For each step $k \geq n_o$ and each queried output point
$y \in \mathbb{R}^{n_y}$, there exists a binary routine
$
\mathrm{OutsideCert}_k(y) \in \{0,1\}
$
such that
$
\mathrm{OutsideCert}_k(y)=1
\;\Longrightarrow\;
y \notin \mathcal{Y}_k.
$
That is, whenever the routine returns $1$, the queried point is certified
to lie outside the true reachable set.
The routine may be conservative in the sense that
$\mathrm{OutsideCert}_k(y)=0$ does not necessarily imply
$y \in \mathcal{Y}_k$.
Its implementation is application-dependent and is not part of the
proposed method.
\end{assumption}

Let $\hat{\mathcal{Y}}_k^{\mathrm{DD}} = \zono{c_k^{\mathrm{DD}}, G_k^{\mathrm{DD}}}$ with generators $G_k^{\mathrm{DD}} = [g_{k,1}^{\mathrm{DD}},\dots,g_{k,\gamma_k}^{\mathrm{DD}}]$. For each nonzero generator, define the unit direction $q_{k,j} = g_{k,j}^{\mathrm{DD}} / \|g_{k,j}^{\mathrm{DD}}\|_2$ and the half-width $\rho_{k,j}^{\mathrm{DD}} = \sum_{\ell=1}^{\gamma_k} |q_{k,j}^\top g_{k,\ell}^{\mathrm{DD}}|$.

To tighten $\hat{\mathcal{Y}}_k^{\mathrm{DD}}$, we query the certificate along the two rays $\pm q_{k,j}$ from $c_k^{\mathrm{DD}}$. For each $j=1,\dots,\gamma_k$, let
\begin{align*}
r_{k,j}^{\pm} = \inf\!\Big\{ r \in [0,\rho_{k,j}^{\mathrm{DD}}] : \mathrm{OutsideCert}_k(c_k^{\mathrm{DD}} \pm r\, q_{k,j}) = 1 \Big\},
\end{align*}
with the convention that the infimum equals $\rho_{k,j}^{\mathrm{DD}}$ if the set is empty.
We then define the certified admissible radius
\begin{equation}\label{eq:rho-certified}
\rho_{k,j}^{\mathrm{cert}}
=
\min\{r_{k,j}^{+},\,r_{k,j}^{-}\},
\qquad j=1,\dots,\gamma_k,
\end{equation}
and the associated contraction factors
\begin{equation}\label{eq:lambda-tight}
\lambda_{k,j}
=
\frac{\rho_{k,j}^{\mathrm{cert}}}{\rho_{k,j}^{\mathrm{DD}}},
\qquad j=1,\dots,\gamma_k.
\end{equation}
Since $0 \le \rho_{k,j}^{\mathrm{cert}} \le \rho_{k,j}^{\mathrm{DD}}$,
we have $0 \le \lambda_{k,j} \le 1$.
The tightened zonotope is defined as
\begin{equation}\label{eq:tight-safe}
\mathcal{T}_k = \zono{c_k^{\mathrm{DD}},\;[\lambda_{k,1} g_{k,1}^{\mathrm{DD}},\;\dots,\;\lambda_{k,\gamma_k} g_{k,\gamma_k}^{\mathrm{DD}}]}.
\end{equation}

\begin{theorem}[Safe generator-contraction tightening]\label{thm:tight}
Let $\hat{\mathcal{Y}}_k^{\mathrm{DD}}=\zono{c_k^{\mathrm{DD}},G_k^{\mathrm{DD}}}$
be the data-driven reachable set from Proposition~\ref{prop:reach},
and let $\mathcal{T}_k$ be defined by~\eqref{eq:rho-certified}--\eqref{eq:tight-safe}.
Then
$
\mathcal{T}_k \subseteq \hat{\mathcal{Y}}_k^{\mathrm{DD}}.
$
\end{theorem}

\begin{proof}
For any $y \in \mathcal{T}_k$, there exists $\alpha \in [-1,1]^{\gamma_k}$ with $y = c_k^{\mathrm{DD}} + \sum_{j=1}^{\gamma_k} \alpha_j \lambda_{k,j} g_{k,j}^{\mathrm{DD}}$. Setting $\beta_j = \alpha_j \lambda_{k,j}$ and using $0 \le \lambda_{k,j} \le 1$, $|\alpha_j| \le 1$, we get $|\beta_j| \le 1$, so $y = c_k^{\mathrm{DD}} + \sum_j \beta_j g_{k,j}^{\mathrm{DD}} \in \hat{\mathcal{Y}}_k^{\mathrm{DD}}$.
\end{proof}

% \begin{remark}[Scope and role of the exterior certificate]\label{rem:tight}
% Three aspects of the certificate-assisted label refinement deserve emphasis.

% First, $\mathrm{OutsideCert}_k$ is a training-time resource only. It is used offline to produce the tightened labels $\mathcal{T}_k$ that the Transformer is trained on. At deployment, the Transformer and the conformal calibration layer operate without any access to the certificate.

% Second, Theorem~\ref{thm:tight} guarantees $\mathcal{T}_k \subseteq \hat{\mathcal{Y}}_k^{\mathrm{DD}}$ but does not guarantee $\mathcal{Y}_k \subseteq \mathcal{T}_k$. The tightened labels are not claimed to be outer approximations of the true reachable set. Their purpose is to provide the Transformer with training targets closer to the true reachable set than the conservative formal envelope.

% Third, the framework yields a three-layer guarantee hierarchy: (i) the formal data-driven outer approximation from Proposition~\ref{prop:reach}, which holds unconditionally; (ii) the certificate-assisted label refinement, which improves training quality when exterior information is available; and (iii) conformal calibration (Theorem~\ref{thm:coverage}), which restores a distribution-free trajectory-point coverage guarantee at deployment using only observable data.
% \hfill $\lrcorner$
% \end{remark}

\begin{remark}[Role of the exterior certificate]\label{rem:tight}
The certificate is used offline only to produce tighter training labels; Theorem~\ref{thm:tight} guarantees $\mathcal{T}_k \subseteq \hat{\mathcal{Y}}_k^{\mathrm{DD}}$ but not $\mathcal{Y}_k \subseteq \mathcal{T}_k$. If the certificate is imperfect, conformal calibration in Theorem~\ref{thm:coverage} compensates at deployment, as no test-time guarantee depends on it. Without tightening, the only available labels are the conservative sets from Algorithm~\ref{alg:dd-reach}, on which a surrogate cannot improve, rendering the Transformer redundant.
\hfill $\lrcorner$
\end{remark}

\subsection{Zonotope Tokenization and Training Data}

Each tightened zonotope $\mathcal{T}_k = \zono{c_k, G_k}$, after reduction to $K_g = \rho_{\max} n_y$ generators, is tokenized as $1 + K_g$ tokens in $\R^{n_y + 1}$: the first encodes the center and the rest encode the generators, each augmented with a normalized time coordinate $t_k / T_{\max}$. Training segments consist of $n_o$ consecutive tightened zonotopes as context and the next zonotope as target, learning $[\mathcal{T}_{k-n_o}, \dots, \mathcal{T}_{k-1}] \to \mathcal{T}_k$. For each of $N_s$ samples, the initial condition is randomly perturbed; tightened labels are produced offline via Section~\ref{sec:tightening}.

\subsection{Zonotope Sequence Prediction}

Reachable-set propagation is inherently causal: the output set $\hat{\mathcal{Y}}_k$ depends only on the preceding sets $\hat{\mathcal{Y}}_{k-1}, \dots, \hat{\mathcal{Y}}_{k-n_o}$ and the corresponding inputs. This causal structure motivates a decoder-only Transformer architecture~\cite{vaswani2017attention}, in which each token attends exclusively to its predecessors through a causal self-attention mask, naturally reflecting the sequential dependency of the propagation.

Since each zonotope $\mathcal{T}_k = \zono{c_k, G_k}$ with $K_g$ generators is decomposed into $1 + K_g$ tokens (one center token followed by $K_g$ generator tokens, each in $\R^{n_y+1}$), a context window of $n_o$ consecutive zonotopes yields a prompt of $n_o(1 + K_g)$ tokens. Each token is projected to the model dimension $d$ through a linear embedding. A stack of $L$ Transformer layers with causal self-attention processes the prompt, and a learned query sequence of $1 + K_g$ tokens is appended to produce the predicted zonotope $\hat{\mathcal{Y}}_k$; the output head projects the final hidden states back to $\R^{n_y}$, reconstructing the center and generators. The training objective minimizes the mean squared error between the predicted and target zonotope tokens, treating center and generator tokens equally:
\[
\mathcal{L} = \frac{1}{1+K_g} \sum_{j=0}^{K_g} \|\hat{t}_j - t_j\|_2^2,
\]
where $t_0$ is the target center token and $t_1, \dots, t_{K_g}$ are the target generator tokens.

At inference, the Transformer is applied autoregressively: the predicted zonotope replaces the oldest context element, and the procedure repeats for each subsequent step. Although the data-driven propagation of Algorithm~\ref{alg:dd-reach} suffers from the conservatism discussed in Remark~\ref{rem:complexity}, the zonotope centers are more trustworthy than the generators, because center propagation uses only the nominal parameter matrix and does not accumulate the generator-level parametric uncertainty. We therefore initialize the autoregressive context with the data-driven centers, providing a reliable starting point for the Transformer, which then predicts tighter generator bounds through its learned representation. By mapping a fixed-length context directly to the next output set, the Transformer sidesteps the structural conservatism of matrix-zonotope multiplication entirely.

\subsection{Conformal Calibration}\label{sec:conformal}

The formal containment guarantee established in Proposition~\ref{prop:reach}
applies to the data-driven reachable sets produced by
Algorithm~\ref{alg:dd-reach}, but not to the Transformer predictions.
To restore a distribution-free coverage guarantee for the learned surrogate,
we apply split conformal prediction using realized output trajectories as
calibration targets. Importantly, the calibration procedure relies only on
observable outputs and the predicted zonotopes, and does not require access
to the true reachable-set hull.

For each calibration trial $i = 1, \dots, n_{\mathrm{cal}}$, a trajectory is
generated from the system and the trained Transformer is run autoregressively
to produce a predicted zonotope at step $k$,
\[
\hat{\mathcal Y}_i^{(k)}
=
\langle \hat c_i^{(k)}, \hat G_i^{(k)} \rangle
=
\left\{
\hat c_i^{(k)} + \hat G_i^{(k)} \beta :
\|\beta\|_\infty \le 1
\right\},
\]
where $\hat c_i^{(k)} \in \mathbb R^{n_y}$ is the zonotope center and
$\hat G_i^{(k)} \in \mathbb R^{n_y \times K_g}$ is the generator matrix.

To conformally calibrate the zonotopic prediction itself, we define the
nonconformity score as the minimum $\ell_\infty$ inflation radius required
for the realized output $y_i^{(k)}$ to belong to the predicted zonotope:
\begin{equation}\label{eq:score}
s_i^{(k)}
=
\inf\Big\{
r \ge 0 :
y_i^{(k)} \in
\hat{\mathcal Y}_i^{(k)} \oplus \zono{0, r I_{n_y}}
\Big\}.
\end{equation}
Equivalently, since
\[
\begin{gathered}
y_i^{(k)} \in
\hat{\mathcal Y}_i^{(k)} \oplus \zono{0, r I_{n_y}} \\
\Updownarrow \\
\exists \beta \in [-1,1]^{K_g}
\ \text{s.t.}\
\|y_i^{(k)}-\hat c_i^{(k)}-\hat G_i^{(k)}\beta\|_\infty \le r .
\end{gathered}
\]
the score admits the optimization form
\begin{equation}\label{eq:score_opt}
s_i^{(k)}
=
\min_{\|\beta\|_\infty \le 1}
\left\|
y_i^{(k)} - \hat c_i^{(k)} - \hat G_i^{(k)} \beta
\right\|_\infty.
\end{equation}
Thus, $s_i^{(k)}$ is precisely the smallest axis-aligned inflation radius
that makes the realized output compatible with the predicted zonotope.

In practice, \eqref{eq:score_opt} can be evaluated as the linear program
\begin{equation}\label{eq:score_lp}
\min_{\beta,t}\; t
\quad\text{s.t.}\quad
\|y_i^{(k)} {-} \hat c_i^{(k)} {-} \hat G_i^{(k)}\beta\|_\infty \le t,\;
\|\beta\|_\infty \le 1.
\end{equation}
Since this optimization is solved only for the offline calibration samples,
it does not affect the online inference cost.

For each horizon step $k$, the conformal quantile is defined by
\begin{equation}\label{eq:qhat}
\hat q^{(k)}
=
\mathrm{Quantile}\bigg(
\{s_i^{(k)}\}_{i=1}^{n_{\mathrm{cal}}},
\frac{\lceil (n_{\mathrm{cal}} + 1)(1-\delta)\rceil}{n_{\mathrm{cal}}}
\bigg).
\end{equation}
At inference time, the conformally calibrated prediction set is then
\begin{equation}\label{eq:inflate}
\hat{\mathcal Y}_k^{(\delta)}
=
\hat{\mathcal Y}_k \oplus \zono{0, \hat q^{(k)} I_{n_y}}.
\end{equation}
For notational simplicity, for the held-out test trajectory we write
$
\hat{\mathcal Y}_k = \hat{\mathcal Y}_{n_{\mathrm{cal}}+1}^{(k)}.
$

The score in~\eqref{eq:score} is computed only during offline calibration.
At test time, one only needs the precomputed quantile $\hat q^{(k)}$ and the
Transformer-predicted zonotope $\hat{\mathcal Y}_k$, so the online correction
amounts to a simple uniform inflation of the prediction by an axis-aligned
box zonotope.

\begin{theorem}[Trajectory-point coverage]\label{thm:coverage}
Let the calibration trajectories
$\{y_i^{(k)}\}_{i=1}^{n_{\mathrm{cal}}}$ and the test trajectory
$y_{n_{\mathrm{cal}}+1}^{(k)}$ be exchangeable. Then the conformally
calibrated set~\eqref{eq:inflate} satisfies
\begin{equation}
\Pr\!\left[
y_{n_{\mathrm{cal}}+1}^{(k)} \in \hat{\mathcal Y}_k^{(\delta)}
\right]
\ge 1-\delta.
\end{equation}
\end{theorem}

\begin{proof}
For a fixed prediction step $k$, each score $s_i^{(k)}$ is computed from the
pair consisting of the realized output and the corresponding Transformer
prediction for trajectory $i$. Under the assumed exchangeability of the
calibration trajectories and the test trajectory, the scores
$\{s_i^{(k)}\}_{i=1}^{n_{\mathrm{cal}}+1}$ are therefore exchangeable.
By the standard split conformal guarantee
\cite[Theorem~2.2]{vovk2005algorithmic},
$
\Pr\!\left[
s_{n_{\mathrm{cal}}+1}^{(k)} \le \hat q^{(k)}
\right]
\ge 1-\delta.
$

By the definition of the score in~\eqref{eq:score},
\[
s_{n_{\mathrm{cal}}+1}^{(k)} \le \hat q^{(k)}
\quad \Longleftrightarrow \quad
y_{n_{\mathrm{cal}}+1}^{(k)}
\in
\hat{\mathcal Y}_{n_{\mathrm{cal}}+1}^{(k)}
\oplus
\zono{0,\hat q^{(k)} I_{n_y}}.
\]
Using the notation
$\hat{\mathcal Y}_k = \hat{\mathcal Y}_{n_{\mathrm{cal}}+1}^{(k)}$,
the right-hand side is exactly
$
y_{n_{\mathrm{cal}}+1}^{(k)} \in \hat{\mathcal Y}_k^{(\delta)}.
$
Hence,
$
\Pr\!\left[
y_{n_{\mathrm{cal}}+1}^{(k)} \in \hat{\mathcal Y}_k^{(\delta)}
\right]
\ge 1-\delta,
$
which proves the claim.
\end{proof}

\begin{corollary}[Trajectory-level coverage]\label{cor:joint}
Define the trajectory-level nonconformity score
$\bar{s}_i = \max_{k} s_i^{(k)}$
and the corresponding quantile
\begin{equation}\label{eq:qbar}
\bar{q} = \mathrm{Quantile}\!\left(
\{\bar{s}_i\}_{i=1}^{n_{\mathrm{cal}}},\;
\frac{\lceil (n_{\mathrm{cal}} + 1)(1-\delta)\rceil}{n_{\mathrm{cal}}}
\right).
\end{equation}
Under the same exchangeability assumption as Theorem~\ref{thm:coverage}, the uniformly inflated sets
$\hat{\mathcal Y}_k^{(\delta,\mathrm{joint})} = \hat{\mathcal Y}_k \oplus \zono{0, \bar{q}\, I_{n_y}}$
satisfy
\begin{equation}
\Pr\!\left[
\forall\, k:\;
y_{n_{\mathrm{cal}}+1}^{(k)} \in \hat{\mathcal Y}_k^{(\delta,\mathrm{joint})}
\right]
\ge 1-\delta.
\end{equation}
\end{corollary}

\begin{proof}
By exchangeability of trajectories, the scores
$\{\bar{s}_i\}_{i=1}^{n_{\mathrm{cal}}+1}$ are exchangeable.
The split conformal guarantee gives
$\Pr[\bar{s}_{n_{\mathrm{cal}}+1} \le \bar{q}] \ge 1-\delta$.
Since $\bar{s}_{n_{\mathrm{cal}}+1} \le \bar{q}$ implies
$s_{n_{\mathrm{cal}}+1}^{(k)} \le \bar{q}$ for every $k$,
the containment
$y_{n_{\mathrm{cal}}+1}^{(k)} \in \hat{\mathcal Y}_k^{(\delta,\mathrm{joint})}$
holds simultaneously for all prediction steps.
\end{proof}

\begin{remark}[Scope of guarantees]\label{rem:scope}
Proposition~\ref{prop:reach} is a deterministic outer approximation for all admissible inputs, disturbances, and initial conditions. Theorem~\ref{thm:coverage} and Corollary~\ref{cor:joint} add distribution-free per-step and trajectory-level coverage guarantees, respectively. The deterministic envelope serves as the safety certificate; the conformally calibrated Transformer provides tighter predictions with quantifiable coverage. Exchangeability is preserved by applying the same filtering criterion to calibration and test trajectories.
\hfill $\lrcorner$
\end{remark}

The complete procedure is summarized in Algorithm~\ref{alg:tf-reach}.

\begin{algorithm}[t]
\caption{Transformer-Enhanced Output Reachability}
\label{alg:tf-reach}
\begin{algorithmic}[1]
\Require Trajectories $\{y_{(k)}^{(t)}\}$; trained model $f_\theta$; quantile $\hat{q}$; context length $n_o$
\Ensure Calibrated output sets $\hat{\mathcal{Y}}_0^{(\delta)}, \dots, \hat{\mathcal{Y}}_N^{(\delta)}$
\For{$k = 0, \dots, n_o - 1$}
    \State $\mathcal{T}_k \gets \mathrm{ZonoFit}(\{y_{(k)}^{(t)}\}, K_g)$ 
\EndFor
\For{$k = n_o, \dots, N$}
    \State Context: $\mathbf{x} \gets \mathrm{Tok}(\mathcal{T}_{k-n_o}, \dots, \mathcal{T}_{k-1})$
    \State Predict: $\hat{\mathcal{Y}}_{k} \gets \mathrm{DeTok}(f_\theta(\mathbf{x}))$
    \State Inflate: $\hat{\mathcal{Y}}_{k}^{(\delta)} \gets \hat{\mathcal{Y}}_{k} \oplus \zono{0, \hat{q}^{(k)} I_{n_y}}$ 
    \State $\mathcal{T}_{k} \gets \hat{\mathcal{Y}}_{k}^{(\delta)}$ 
\EndFor
\end{algorithmic}
\end{algorithm}

\section{Numerical Simulations}\label{sec:experiments}

\begin{figure*}[!h]
    \vspace{-1em}
    \centering
    \begin{subfigure}[h]{0.32\textwidth}
        \includegraphics[width=\linewidth]{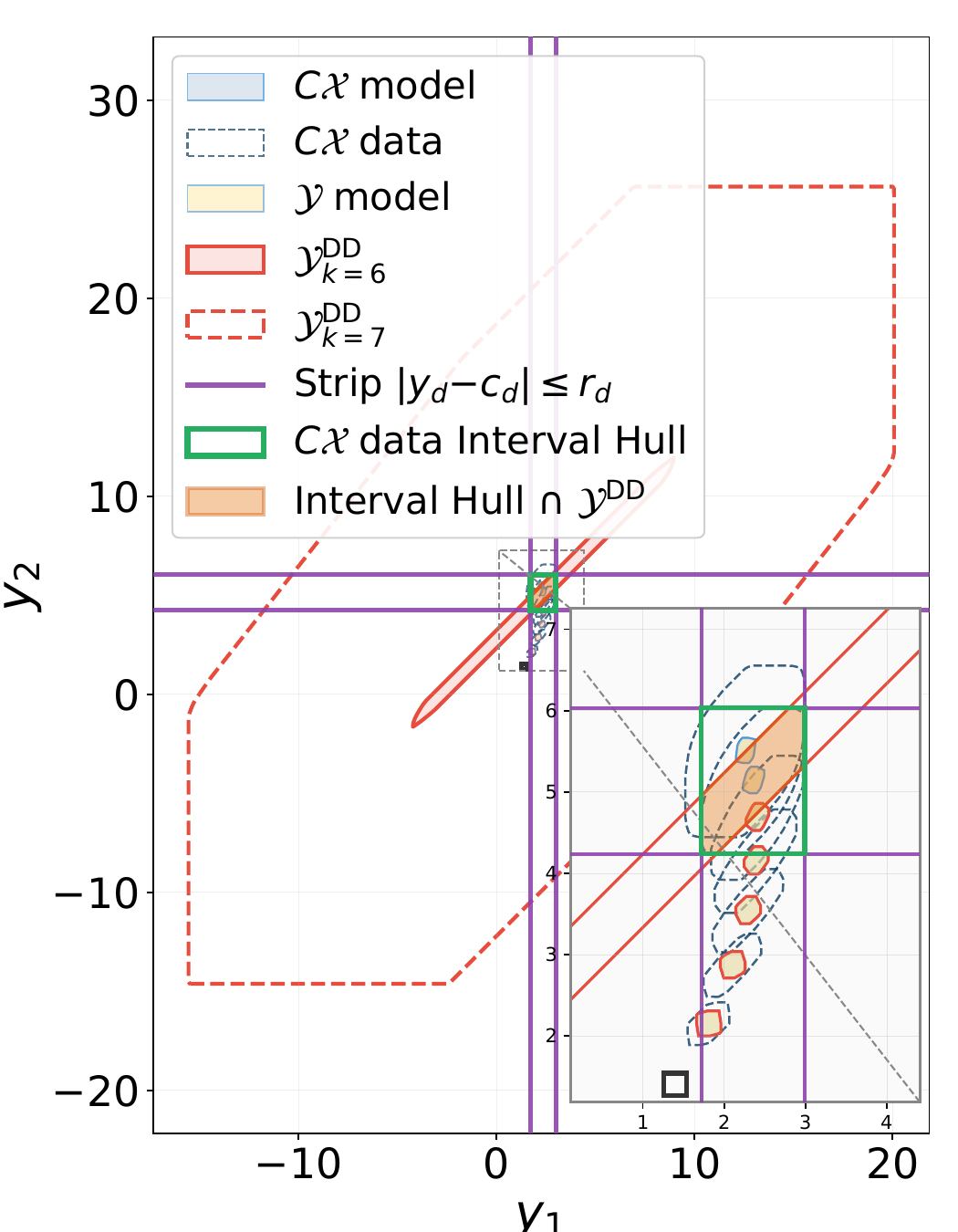}
        \caption{$C_a$: cross-block}
        \label{fig:dd_Ca}
    \end{subfigure}
    \begin{subfigure}[h]{0.32\textwidth}
        \includegraphics[width=\linewidth]{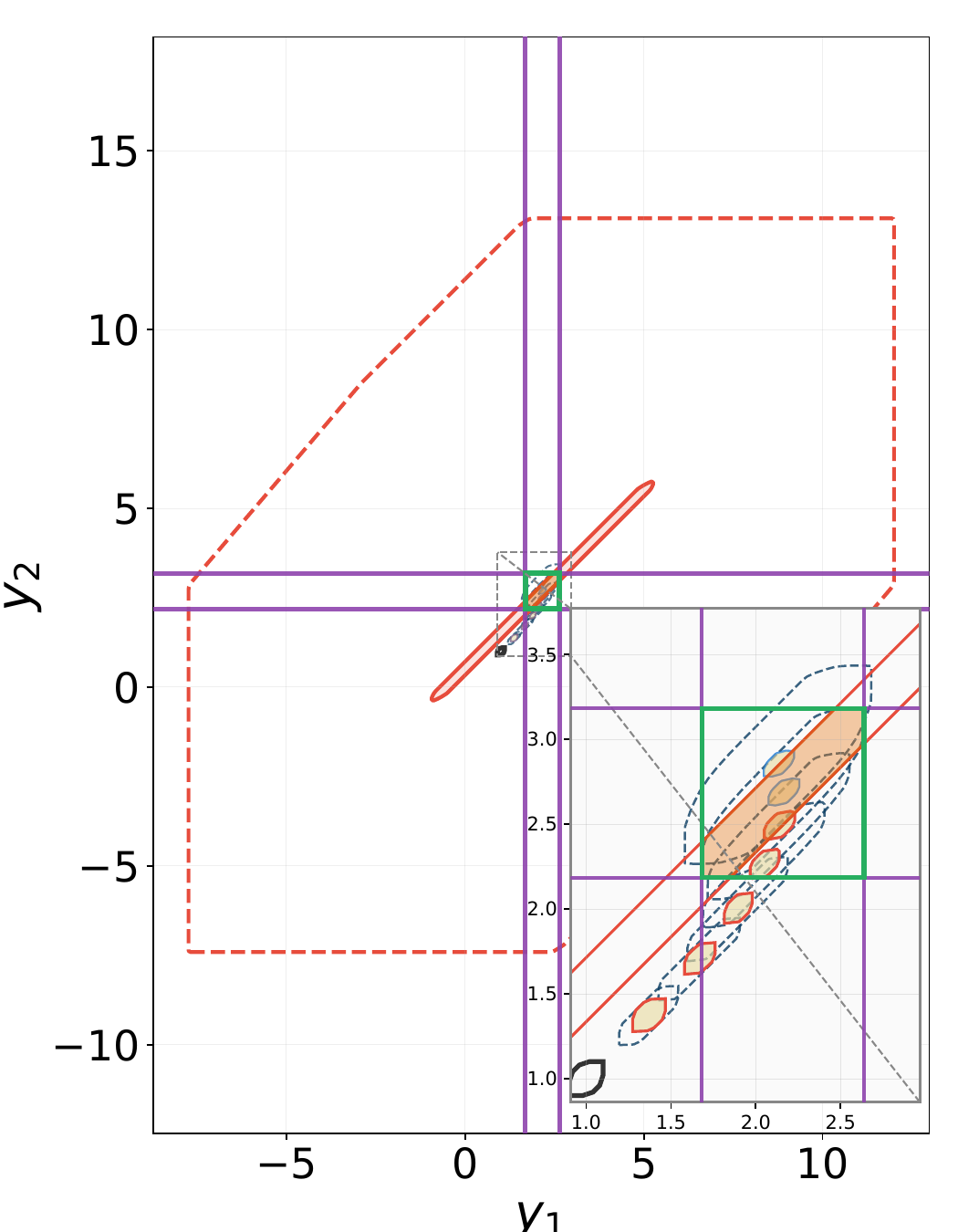}
        \caption{$C_b$: gradient}
        \label{fig:dd_Cb}
    \end{subfigure}
    \begin{subfigure}[h]{0.32\textwidth}
        \includegraphics[width=\linewidth]{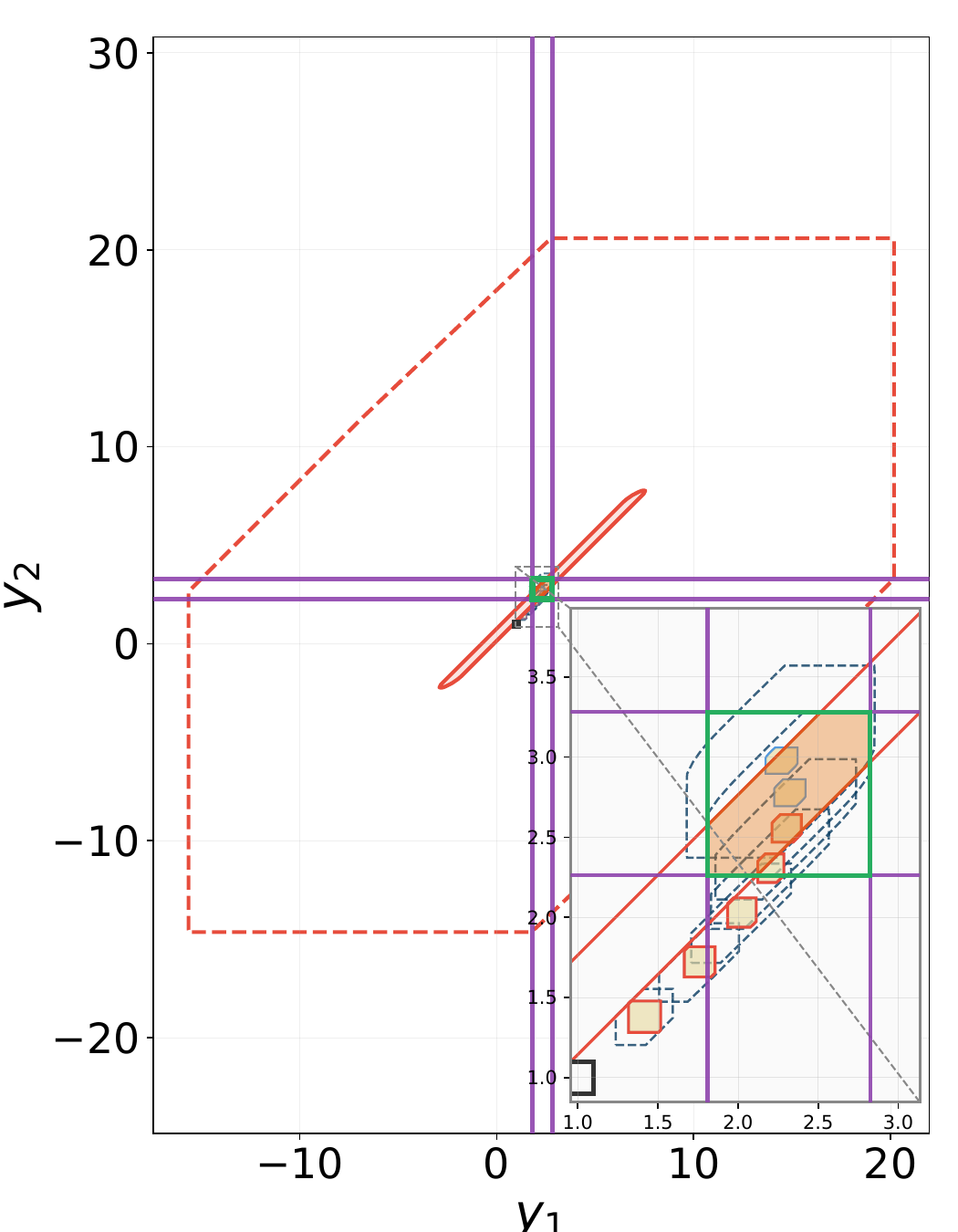}
        \caption{$C_c$: pairwise}
        \label{fig:dd_Cc}
    \end{subfigure}
    \caption{Data-driven output reachable sets under three sensor configurations ($n_y=2$).}
    \label{fig:dd-partial}
\end{figure*}

% ── Figure: TF results ──
\begin{figure*}[!h]
    \centering
    \begin{subfigure}[h]{0.32\textwidth}
        \includegraphics[width=\linewidth]{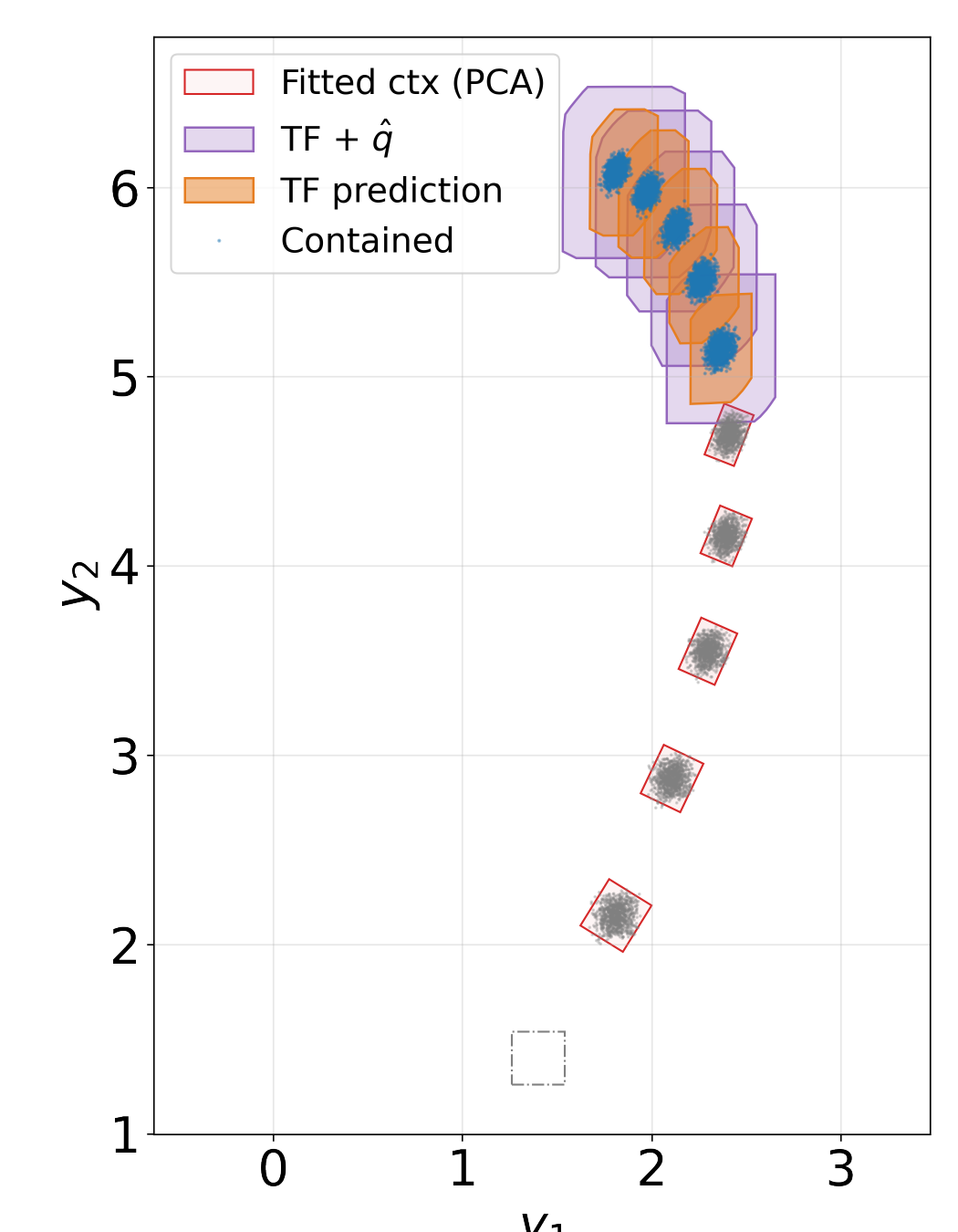}
        \caption{$C_a$: cross-block}
        \label{fig:tf_Ca}
    \end{subfigure}
    \begin{subfigure}[h]{0.32\textwidth}
        \includegraphics[width=\linewidth]{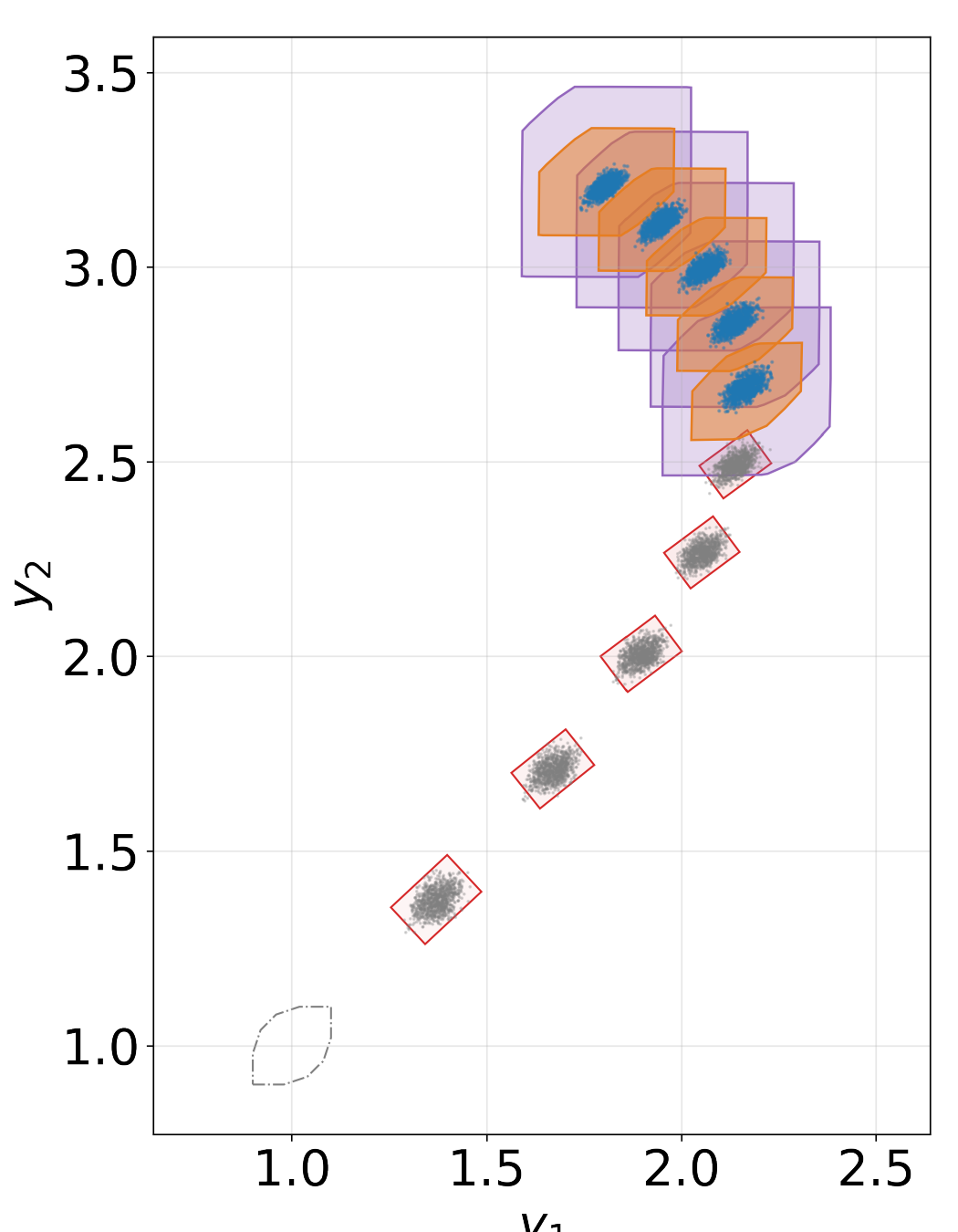}
        \caption{$C_b$: gradient}
        \label{fig:tf_Cb}
    \end{subfigure}
    \begin{subfigure}[h]{0.32\textwidth}
        \includegraphics[width=\linewidth]{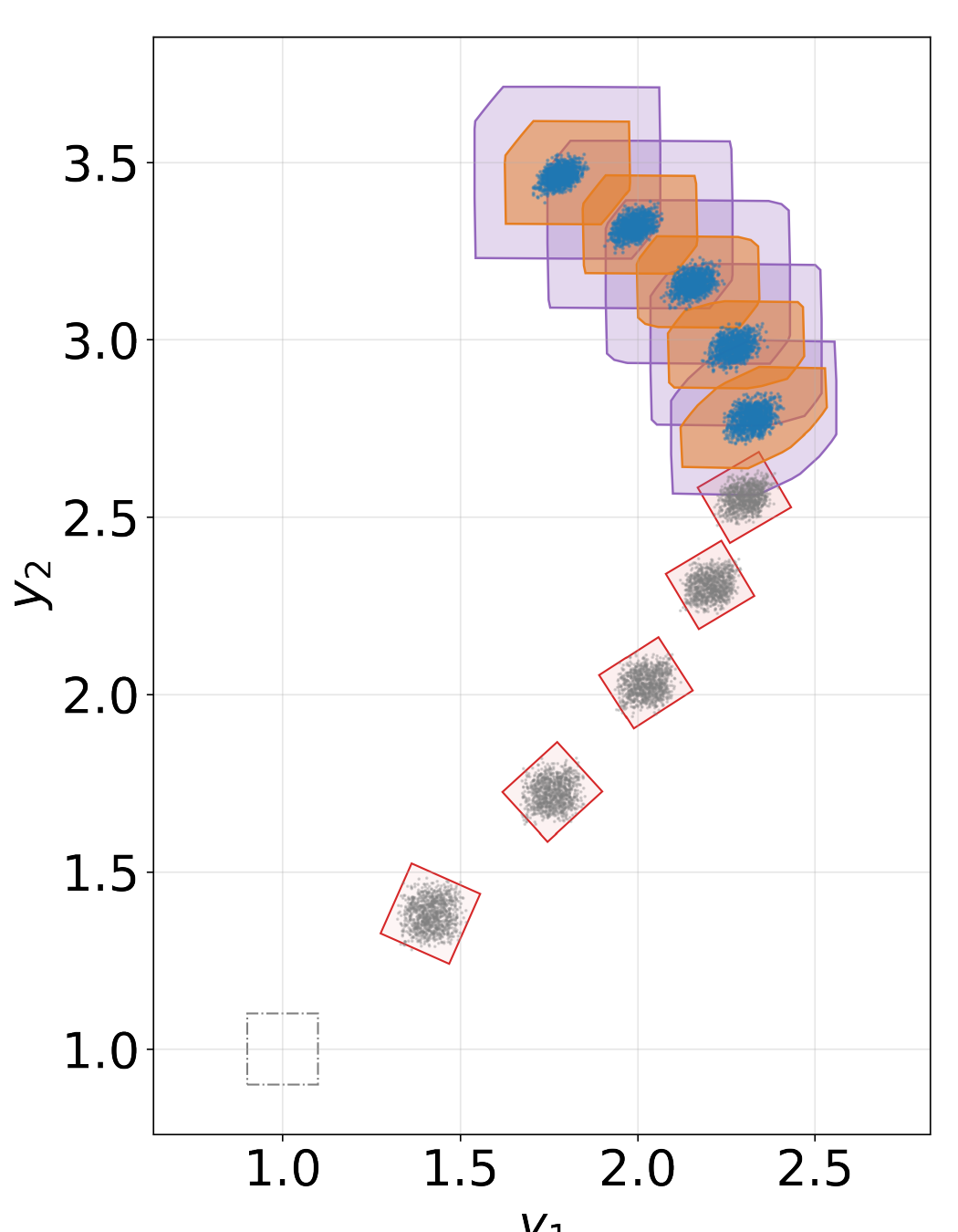}
        \caption{$C_c$: pairwise}
        \label{fig:tf_Cc}
    \end{subfigure}
    \caption{Transformer-enhanced output reachable sets for three sensor configurations.}
    \label{fig:tf-results}
\end{figure*}

We validate the proposed framework on a five-dimensional LTI system~\eqref{eq:ss} with $n_x=5$, $n_u=1$, and sampling time $\Delta t=0.05\,\mathrm{s}$, where $A \in \R^{5\times 5}$ is the oscillatory matrix from~\cite{alanwar2023data} and $B=\mathbf{1}_{5\times 1}$. The system matrices $(A, B, C)$ are unknown to the algorithm and are used solely to generate simulation data and to construct a model-based reference for evaluation. The proposed method receives only noisy input-output measurements $\{(u_{(k)},y_{(k)})\}$, the system order $n_o = n_x$, and the residual bound $\mathcal{Z}_\varepsilon$. The initial state set is $\mathcal{X}_0=\zono{1_{5\times1},0.1I_5}$, the input set is $\mathcal{U}=\zono{10,0.25}$, and the aggregated residual bound is $\mathcal{Z}_\varepsilon=\zono{0,0.006\cdot \mathbf{1}_{n_y\times1}}$, consistent with Assumption~\ref{ass:residual-bound}.

% \subsection{Data-Driven Output Reachability}
The model set $\mzon_\Sigma$ is constructed from $T = 50$ input-output samples using Lemma~\ref{lm:model_set} with system order $n_o = 5$ and reduction order $\rho_{\max} = 200$. To demonstrate that the conservatism of Algorithm~\ref{alg:dd-reach} is structural rather than tied to a specific sensor configuration, we evaluate three non-trivial output matrices of dimension $n_y = 2$, each satisfying the observability rank condition: the cross-block configuration $C_a = [0.6, 0, 0.8, 0, 0;\; 0, 0.8, 0, 0, 0.6]$ fuses measurements from dynamically decoupled subsystems, the gradient configuration $C_b = [0.4, 0.3, 0.2, 0.1, 0;\; 0, 0.1, 0.2, 0.3, 0.4]$ assigns monotonically varying weights, and the pairwise configuration $C_c = [0.5, 0.5, 0, 0, 0;\; 0, 0, 0.3, 0, 0.7]$ averages non-adjacent state pairs. In all three cases, the lifted regressor dimension is $n_o(n_y + n_u) = 15$.

Fig.~\ref{fig:dd-partial} shows the propagated output reachable sets for each sensor configuration. The model-based state reachable set projected through $C$ and the model-based output reachable set $\mathcal{Y}$ model serve as references. The data-driven output sets $\mathcal{Y}^{\mathrm{DD}}_k$ are plotted at steps $k = n_o{+}1$ and $k = n_o{+}2$; the interval hull of the data-driven state projection $C\mathcal{X}$ data is also shown. At the first prediction step, the DD envelope expands by an order of magnitude relative to the model-based reference; at $k = n_o{+}2$ it grows further, confirming that the conservatism is structural and independent of the sensor configuration.

As a concrete instantiation of Assumption~\ref{ass:directional-certificate}, we define the exterior certificate along the coordinate directions $e_d$, $d=1,\dots,n_y$: from historical output trajectories, empirical bounds $c_d \pm r_d$ are estimated for each output dimension, and $\mathrm{OutsideCert}_k(y) = 1$ whenever $|y_d - c_d| > r_d$ for some $d$. Applying the directional contraction of Section~\ref{sec:tightening} with these certificates yields tightened sets that substantially reduce the DD envelope without requiring knowledge of $C$.

% \subsection{Transformer-Enhanced Output Reachability}

The Transformer has model dimension $d = 128$, $8$ attention heads, $4$ decoder layers, and $d_{\text{ff}} = 512$. Training labels are tightened by strip-based constraints $|y_d - c_d| \le r_d$ derived from historical trajectories, without requiring knowledge of $C$. Training uses $N_s = 10{,}000$ samples with initial-condition augmentation, learning rate $3{\times}10^{-4}$, patience $150$, and $1500$ epochs. The context spans $n_o = 5$ zonotopes with $K_g = 8$ generators, yielding $45$ input tokens; the target is the next zonotope with $9$ tokens. At inference, context zonotopes are fitted from $M = 500$ trajectories via Lemma~\ref{lm:pca-fit}, and the Transformer predicts autoregressively for $5$ steps. Conformal calibration uses $n_{\mathrm{cal}} = 200$ trajectories at $\delta = 0.05$; scores~\eqref{eq:score} are computed against realized outputs, and the quantile~\eqref{eq:qhat} inflates each set.

Fig.~\ref{fig:tf-results} shows results for all three configurations. The raw Transformer predictions extend $5$ steps beyond the observation horizon, and the conformally calibrated predictions TF$+\hat{q}$ provide inflated sets with coverage guarantees. Of $1000$ test trajectories, those with initial $n_o$ outputs inside the context zonotopes are retained, and all retained trajectories achieve $100\%$ containment at every step. Table~\ref{tab:baseline} reports the mean interval hull width and empirical coverage. The SysID baseline for $C_b$ and $C_c$ yields spectrally unstable models ($\dagger$), so its apparent tightness carries no safety certificate. The conformally calibrated Transformer achieves $100\%$ coverage across all configurations, with mean interval width $1.5$--$1.7\times$ the Model$+C$ baseline and an order-of-magnitude reduction over DD.

% ── Figure: DD partial observability ──

% ── Table: Baseline comparison ──
\begin{table}[!t]
    \centering
    \caption{Mean interval hull width and empirical coverage at prediction steps $n_o{+}1$ to $n_o{+}3$. MC: Monte Carlo hull; Model$+C$: model-based with known $C$; TF$+\hat{q}$: conformal Transformer; DD: Algorithm~\ref{alg:dd-reach}; SysID: Ho--Kalman identification. $\dagger$: spectrally unstable identified model.}
    \label{tab:baseline}
    \renewcommand{\arraystretch}{1.1}
    \begin{tabular}{c c c c c c c c}
        \hline
         & Step & MC & Model$+C$ & TF$+\hat{q}$ & DD & SysID & Cov. \\
        \hline
        \multirow{3}{*}{$C_a$}
         & $6$ & $0.254$ & $0.451$ & $0.680$ & $3.298$ & $0.371$ & $100\%$ \\
         & $7$ & $0.232$ & $0.443$ & $0.689$ & $4.558$ & $0.376$ & $100\%$ \\
         & $8$ & $0.219$ & $0.442$ & $0.729$ & $6.256$ & $0.441$ & $100\%$ \\
        \hline
        \multirow{3}{*}{$C_b$}
         & $6$ & $0.132$ & $0.281$ & $0.423$ & $2.048$ & $0.098^\dagger$ & $100\%$ \\
         & $7$ & $0.134$ & $0.283$ & $0.437$ & $2.790$ & $0.084^\dagger$ & $100\%$ \\
         & $8$ & $0.124$ & $0.288$ & $0.464$ & $3.783$ & $0.095^\dagger$ & $100\%$ \\
        \hline
        \multirow{3}{*}{$C_c$}
         & $6$ & $0.168$ & $0.296$ & $0.448$ & $2.096$ & $0.120^\dagger$ & $100\%$ \\
         & $7$ & $0.160$ & $0.304$ & $0.475$ & $2.866$ & $0.089^\dagger$ & $100\%$ \\
         & $8$ & $0.151$ & $0.314$ & $0.487$ & $3.904$ & $0.105^\dagger$ & $100\%$ \\
        \hline
    \end{tabular}
\end{table}

\section{Conclusion}\label{sec:conclusion}

We presented a data-driven framework for output reachability analysis of LTI systems with unknown dynamics. The Cayley--Hamilton theorem is applied to eliminate the latent state, yielding an autoregressive model whose parameter uncertainty is encoded as a matrix zonotope; formal output reachable sets are then propagated with set-containment guarantees, conditional on a valid outer bound on the aggregated residual. A decoder-only Transformer trained on certificate-assisted tighter labels replaces the conservative sequential propagation with direct set-valued predictions, and split conformal prediction restores distribution-free trajectory-point coverage at deployment. Experiments on a five-dimensional system with multiple unknown output matrices confirm that the conformally calibrated Transformer achieves tight output sets while maintaining the nominal coverage level. Future work will address nonlinear output models and adaptive calibration for non-stationary settings.

\bibliographystyle{IEEEtran}
\bibliography{ref}

@inproceedings{luetzow2023reachability,
  title={Reachability analysis of {ARMAX} models},
  author={L{\"u}tzow, Laura and Althoff, Matthias},
  booktitle={2023 62nd IEEE Conference on Decision and Control (CDC)},
  pages={4157--4164},
  year={2023},
  organization={IEEE}
}

@article{huang2026cddr,
  title={Conformalized data-driven reachability analysis with {PAC} guarantees},
  author={Huang, Yanliang and Zhang, Zhen and Xie, Peng and Zeng, Zhuoqi and Alanwar, Amr},
  journal={arXiv preprint arXiv:2603.12220},
  year={2026}
}

@article{van1994n4sid,
  title={N4SID: Subspace algorithms for the identification of combined deterministic-stochastic systems},
  author={Van Overschee, Peter and De Moor, Bart},
  journal={Automatica},
  volume={30},
  number={1},
  pages={75--93},
  year={1994},
  publisher={Elsevier}
}

@inproceedings{zhang2025data,
  title={Data-driven nonconvex reachability analysis using exact multiplication},
  author={Zhang, Zhen and Niazi, M Umar B and Chong, Michelle S and Johansson, Karl H and Alanwar, Amr},
  booktitle={2025 IEEE 64th Conference on Decision and Control (CDC)},
  pages={4882--4889},
  year={2025},
  organization={IEEE}
}

@article{scott2016constrained,
  title={Constrained zonotopes: A new tool for set-based estimation and fault detection},
  author={Scott, Joseph K and Raimondo, Davide M and Marseglia, Giuseppe Roberto and Braatz, Richard D},
  journal={Automatica},
  volume={69},
  pages={126--136},
  year={2016},
  publisher={Elsevier}
}

@inproceedings{kochdumper2020sparse,
  title={Sparse polynomial zonotopes: A novel set representation for reachability analysis},
  author={Kochdumper, Niklas and Althoff, Matthias},
  booktitle={IEEE Transactions on Automatic Control},
  volume={66},
  number={9},
  pages={4043--4058},
  year={2020},
  publisher={IEEE}
}

@inproceedings{althoff2008reachability,
  title={Reachability analysis of nonlinear systems with uncertain parameters using conservative linearization},
  author={Althoff, Matthias and Stursberg, Olaf and Buss, Martin},
  booktitle={2008 47th IEEE Conference on Decision and Control},
  pages={4042--4048},
  year={2008},
  organization={IEEE}
}

@article{althoff2021set,
  title={Set propagation techniques for reachability analysis},
  author={Althoff, Matthias and Frehse, Goran and Girard, Antoine},
  journal={Annual Review of Control, Robotics, and Autonomous Systems},
  volume={4},
  number={1},
  pages={369--395},
  year={2021},
  publisher={Annual Reviews}
}

@book{rierson2017developing,
  title={Developing safety-critical software: a practical guide for aviation software and DO-178C compliance},
  author={Rierson, Leanna},
  year={2017},
  publisher={CRC Press}
}

@inproceedings{girard2005reachability,
  title={Reachability of uncertain linear systems using zonotopes},
  author={Girard, Antoine},
  booktitle={International workshop on hybrid systems: Computation and control},
  pages={291--305},
  year={2005},
  organization={Springer}
}

@inproceedings{girard2006efficient,
  title={Efficient computation of reachable sets of linear time-invariant systems with inputs},
  author={Girard, Antoine and Le Guernic, Colas and Maler, Oded},
  booktitle={International workshop on hybrid systems: Computation and control},
  pages={257--271},
  year={2006},
  organization={Springer}
}

@inproceedings{alanwar2021data,
  title={Data-driven reachability analysis using matrix zonotopes},
  author={Alanwar, Amr and Koch, Anne and Allg{\"o}wer, Frank and Johansson, Karl Henrik},
  booktitle={Learning for Dynamics and Control},
  pages={163--175},
  year={2021},
  organization={PMLR}
}

@phdthesis{althoff2010reachability,
  title={Reachability analysis and its application to the safety assessment of autonomous cars},
  author={Althoff, Matthias},
  year={2010},
  school={Technische Universit{\"a}t M{\"u}nchen}
}

@article{alanwar2023data,
  title={Data-driven reachability analysis from noisy data},
  author={Alanwar, Amr and Koch, Anne and Allg{\"o}wer, Frank and Johansson, Karl Henrik},
  journal={IEEE Transactions on Automatic Control},
  volume={68},
  number={5},
  pages={3054--3069},
  year={2023},
  publisher={IEEE}
}

@inproceedings{thorpe2021sreachtools,
  title={SReachTools kernel module: Data-driven stochastic reachability using Hilbert space embeddings of distributions},
  author={Thorpe, Adam J and Ortiz, Kendric R and Oishi, Meeko MK},
  booktitle={2021 60th IEEE Conference on Decision and Control (CDC)},
  pages={5073--5079},
  year={2021},
  organization={IEEE}
}

@article{devonport2023data,
  title={Data-driven reachability and support estimation with Christoffel functions},
  author={Devonport, Alex and Yang, Forest and El Ghaoui, Laurent and Arcak, Murat},
  journal={IEEE Transactions on Automatic Control},
  volume={68},
  number={9},
  pages={5216--5229},
  year={2023},
  publisher={IEEE}
}

@inproceedings{griffioen2023data,
  title={Data-driven reachability analysis for Gaussian process state space models},
  author={Griffioen, Paul and Arcak, Murat},
  booktitle={2023 62nd IEEE Conference on Decision and Control (CDC)},
  pages={4100--4105},
  year={2023},
  organization={IEEE}
}

@inproceedings{sivaramakrishnan2024stochastic,
  title={Stochastic reachability of uncontrolled systems via probability measures: Approximation via deep neural networks},
  author={Sivaramakrishnan, Karthik and Sivaramakrishnan, Vignesh and Devonport, Rosalyn A and Oishi, Meeko MK},
  booktitle={2024 IEEE 63rd Conference on Decision and Control (CDC)},
  pages={7534--7541},
  year={2024},
  organization={IEEE}
}

@article{hashemi2024statistical,
  title={Statistical reachability analysis of stochastic cyber-physical systems under distribution shift},
  author={Hashemi, Navid and Lindemann, Lars and Deshmukh, Jyotirmoy V},
  journal={IEEE Transactions on Computer-Aided Design of Integrated Circuits and Systems},
  volume={43},
  number={11},
  pages={4250--4261},
  year={2024},
  publisher={IEEE}
}

@inproceedings{alanwar2022data,
  title={Data-driven set-based estimation using matrix zonotopes with set containment guarantees},
  author={Alanwar, Amr and Berndt, Alexander and Johansson, Karl Henrik and Sandberg, Henrik},
  booktitle={2022 European Control Conference (ECC)},
  pages={875--881},
  year={2022},
  organization={IEEE}
}

@book{horn2012matrix,
  title={Matrix analysis},
  author={Horn, Roger A and Johnson, Charles R},
  year={2012},
  publisher={Cambridge university press}
}

@article{kuhn1998rigorously,
  title={Rigorously computed orbits of dynamical systems without the wrapping effect},
  author={K{\"u}hn, Wolfgang},
  journal={Computing},
  volume={61},
  number={1},
  pages={47--67},
  year={1998},
  publisher={Springer}
}

@article{vaswani2017attention,
  title={Attention is all you need},
  author={Vaswani, Ashish and Shazeer, Noam and Parmar, Niki and Uszkoreit, Jakob and Jones, Llion and Gomez, Aidan N and Kaiser, {\L}ukasz and Polosukhin, Illia},
  journal={Advances in neural information processing systems},
  volume={30},
  year={2017}
}

@book{vovk2005algorithmic,
  title={Algorithmic learning in a random world},
  author={Vovk, Vladimir and Gammerman, Alexander and Shafer, Glenn},
  year={2005},
  publisher={Springer}
}

@article{shafer2008tutorial,
  title={A tutorial on conformal prediction.},
  author={Shafer, Glenn and Vovk, Vladimir},
  journal={Journal of machine learning research},
  volume={9},
  number={3},
  year={2008}
}

@article{milanese2004set,
  title={Set membership identification of nonlinear systems},
  author={Milanese, Mario and Novara, Carlo},
  journal={Automatica},
  volume={40},
  number={6},
  pages={957--975},
  year={2004},
  publisher={Elsevier}
}

@article{combastel2015zonotopes,
  title={Zonotopes and Kalman observers: Gain optimality under distinct uncertainty paradigms and robust convergence},
  author={Combastel, Christophe},
  journal={Automatica},
  volume={55},
  pages={265--273},
  year={2015},
  publisher={Elsevier}
}

@book{hespanha2018linear,
  title={Linear Systems Theory},
  author={Hespanha, Jo{\~a}o P.},
  year={2018},
  edition={2},
  publisher={Princeton University Press}
}

\end{document}